\documentclass[copyright,creativecommons]{eptcs}

\providecommand{\event}{ICE 2010} 

\usepackage{amsmath, amssymb, amsthm}
\usepackage{stmaryrd}
\usepackage{epic,eepic}
\usepackage{graphicx}
\usepackage{graphics}
\usepackage{url}

\input{macros.sty}

\newif\ifapp
\appfalse

\title{A Graphical Approach to Progress for Structured Communication
  in Web Services\thanks{This work is funded in part by the Danish
    Research Agency (grant no.: 2106-080046) and the IT University of
    Copenhagen (the Jingling Genies projects). Authors are listed
    alphabetically by last name.}}

\author{
  Marco Carbone \qquad\qquad\qquad\qquad S\o ren Debois
  \institute{IT University of Copenhagen\\ Copenhagen, Denmark}
  \email{\quad \{carbonem,debois\}@itu.dk}
} 

\def\titlerunning{A Graphical Approach to Progress for Structured
  Communication in Web Services} \def\authorrunning{M. Carbone \&
  S. Debois}

\begin{document}
\maketitle

\begin{abstract}
 We investigate a graphical
representation of session invocation interdependency in order to prove
progress for the $\pi$-calculus with sessions under the usual session
typing discipline.  We show that those processes whose associated
dependency graph
is acyclic can be brought to reduce. We call such processes {\em
  transparent processes}.
Additionally, we prove that for well-typed processes where services contain no
free names, such acyclicity is preserved by the
reduction semantics.

Our results encompass programs (processes containing
neither free nor restricted session channels) and higher-order
sessions (delegation). Furthermore, we give examples suggesting that
transparent processes constitute a large enough class of 
 processes with progress to have 
applications in modern session-based programming languages for web
services.
%
\end{abstract}

\section{Introduction}
\label{sec:introduction}
Due to fast-growing technologies and exponential growth of the Internet
and the world-wide web, computing systems and software based on
communication are becoming the norm rather than the exception. In
particular, {\em Web Services (WS)} is today a crucial ingredient in
many such systems. The W3C, the world-wide web's governing body,
defines WS as {\em ``a software system designed to support
  interoperable machine-to-machine interaction over a network''}
\cite{W3C}.
Abstractly, we think of WS as running processes, each identified by
some name, which can be repeatedly invoked by clients or other
services.  Such an invocation spawns a new thread of the service which
handles the actual interaction between service and client. This
interaction, a sequence of input-output operations, is
often referred to as {\em session}.


Recently, sessions have been the subject of intense research. Most
pertinent to the present paper, calculi for concurrency have been
equipped with typing system ensuring session safety
\cite{carbone.honda.yoshida:esop07,THK,honda.vasconcelos.kubo:language-primitives}. Such
session typing systems have in turn given rise to a host of
programming languages using session types to control concurrency,
e.g., \cite{HU07TYPE-SAFE,scribble,PucellaTov08} to cite a few.  Such
languages derive their strengths and practicality directly from our
theoretical understanding of the underlying calculi.

One class of theoretical problems with direct ramifications for
programming languages is that of ensuring {\em progress} of sessions,
that is, statically ensuring that protocols do not inadvertently {\em
  get stuck} or {\em deadlock}
\cite{sessionMariangiola,CDY07,DLY07,BCDDDY:concur2008}.  In the
present work, we investigate a graphical representation of session
invocation interdependency and exploit its properties for proving
progress. As a result, we identify a new class of processes, which we
call the {\em transparent} processes. This class advances the state of
the art in two directions: (1) it includes processes not hitherto
identified as having progress and (2) it is characterised by a simple
syntactic criterion which requires no special-purpose typing system or
inference and is computable in linear time with respect to the number
of nodes in the abstract syntax tree of a process.  We shall argue
that the class we consider has specific practical motivation and, in
combination with (2), it thus seems to be of potentially direct
practical importance.

In session-based systems such as WS, caller and callee play central
roles. Traditional calculi take a completely symmetric approach to
such roles. We claim that this approach is perhaps not fully aligned
with the practicalities of WS: the communicating parties, clients and
services, are not equal but rather \emph{inherently
  asymmetric}. Whereas clients can be anything, services must be
mutually independent. Thus, we shall assume that \emph{a service can
  never depend on previously opened communications}. This assumption
leads  us to the class of transparent processes, while
simultaneously assuring practical relevance.

Formally, we work in a $\pi$-calculus with sessions and session types
\'a l\`a \cite{honda.vasconcelos.kubo:language-primitives}.  In this
model, the above assumption becomes simply that no service contain
free session channels.  This assumption is by itself sufficient to
guarantee progress, without any extra constraints on well-typed
processes. Additionally, the progress result of the present paper
actually goes beyond such self-contained services i.e. \stable{}
processes are a larger class than just closed services. Technically,
we adopt 
a near-standard session-typing discipline similar to the one used in
\cite{carbone.honda.yoshida:esop07,CHY:concur2008}. That
 services cannot rely on already open communications is reflected
by the following only non-standard typing rule for services. By example:
\begin{align*}
  \Did{T-Serv} \quad
  &
  \Rule
  {
    \tproves
    {\Gamma,\;
      \texttt{buy}:
      \langle\alpha\rangle
    \quad} 
    {\quad\text{\framebox{$P$}}\quad}
    {\text{\framebox{$k:\alpha$}}}
  }
  {
    \tproves
    {\Gamma,\;
      \texttt{buy}:
      \langle\alpha\rangle
    \quad} 
    {\quad\text{\framebox{$\initIn {\texttt{buy}}k\pfx P$}}\quad}
    {\text{\framebox{$\emptyset$}}}
  }
\end{align*}
Above, we have framed those points relevant to the discussion. The
term $\initIn {\texttt{buy}}k\pfx P$ denotes a service named
\texttt{buy} which, upon invocation, will create some private channel
$k$ and use it in its body $P$ for exchanging messages with the
invoker. The typing rule says that the body $P$ has access only to the
new session $k$ (expressed in the premise as $k:\alpha$ where $\alpha$
describes how $k$ will be used in $P$), and not to any other previously
opened session.  Technically, this is just a small restriction to
standard session typing
\cite{honda.vasconcelos.kubo:language-primitives}, thus any process
typeable in the present system will be typeable in the standard
one. On the other hand, this restriction is practically reasonable
\cite{carbone.honda.yoshida:esop07}.



The central idea for proving the progress property is to focus on the
development of {\em session dependency graphs}, first introduced in
\cite{CHY:concur2008}. These capture the key intuition in our
transparent processes. 
Whereas much work in progress on deadlock works essentially by
ordering the sequential use of channels
(e.g. \cite{Kobayashi06,DLY07,brunimezzina2008}), the present work,
\emph{qua\/} the session dependency graph, works instead by ordering
the running threads of a system by their pairwise sharing of sessions.
Consider the following example; we call prefixes in a parallel
composition {\em threads}:
\begin{displaymath}
  \begin{array}{ll}
    \underbrace{k?(x). k'!\langle x\rangle}_{\text{1}}
    \ \pp\
    \underbrace{k!\langle 5\rangle}_{2}
    \ \pp \
    \underbrace{k'?(y)}_3
  \end{array}
\end{displaymath}
Above, thread $1$ receives a value on session channel $k$ and then
forwards it to another session $k'$. Thread $2$ sends value $5$ on
session $k$ while thread $3$ waits for a message to be sent on session
$k'$. Each of the three processes above represents a node in the
session dependency graph. Moreover, an edge between two nodes is in
the graph whenever their corresponding processes share a free session
channel.
Graphically: 
\begin{center}
  \includegraphics{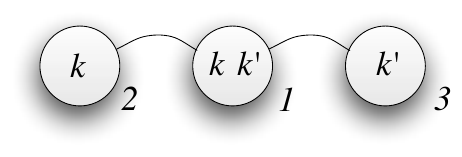}
\end{center}
The \emph{\stable{}} processes are simply those processes which have
acyclic session dependency graph at every sub-term.
Acyclicity at top-level is enough to guarantee the absence of
immediate deadlock; acyclicity at every sub-term ensures that this
deadlock--freedom is preserved by reduction. While the session
dependency graph approximates session interference and interdependency
in a powerful and intuitive way, the acyclicity of the session
dependency graphs is at heart a simple syntactic property, easily
checkable in linear time by considering the free names of a
process. 

\paragraph{Related Work.}
Progress has been investigated for a variety of calculi for web
services e.g., in
\cite{acciaiboreale2008,brunimezzina2008,caires:esop09} and, for
session types, it has spawned several lines of research. The present
paper takes as its starting points \cite{DLY07} and
\cite{CHY:concur2008}. In particular, the former gives a progress
result for a class of well-typed processes identified by a particular
typing system based on finding an ordering of channel usage, in
contrast to the session dependency graph ordering
threads. Importantly, \cite{DLY07} allows service channels to be
restricted, something that we do not presently. 
%
%
However, \stable\ processes are neither a subset nor a superset of the
processes characterized in \cite{DLY07}.  We clarify the key
differences with some examples. 
First, the process
\begin{equation}
\label{eq:2}
  \begin{array}{l}
    \initIn {\texttt{buy}}k\pfx\initOut {\texttt{ship}}{k'}\pfx {k}\triangleright
      \left\{
      \begin{array}{rl}
        \texttt{ok}:   & \inputP k{x_{\texttt{addr}}}\pfx\outputP {k'}{x_{\texttt{addr}}},\\
        \texttt{abort}:& \outputP {k'}{\texttt{null}}\pfx\inputP {k}{x_{\texttt{reason}}}
      \end{array}
    \right\}
\end{array}
\end{equation}
denotes a service $\texttt{buy}$ which, upon invocation, creates the
session channel $k$, then calls service $\texttt{ship}$ and finally
branches (with labels $\texttt{ok}$ and $\texttt{abort}$) on $k$. The
two branches differ by the order in which $k$ and $k'$ are used. This
process is transparent, and thus has progress (in a context with a
suitable invocation $\initOut {\texttt{buy}} k$ and service $\initIn
{\texttt{ship}} {k'}$).  However, because of the inconsistent
orderings on the use of channels in the two branches (communication on
$k$ then $k'$ in one, $k'$ then $k$ in the other), the typing of
\cite{DLY07} rejects this process.
Second, the process
\begin{equation}\label{eq:2asd}
  \initIn {\texttt{buy}}k\pfx \inputP k{x_{\texttt{card}}}\pfx\initIn
  {\texttt {serv}}{k'}\pfx\outputP {k'}5
  \qquad\pp\qquad
  \initOut {\texttt{buy}}k\pfx \initOut {\texttt{serv}}{k'}\pfx\inputP {k'}y\pfx\outputP k{\texttt{card}}
\end{equation}
consists of the parallel composition of a service $\texttt{buy}$ and a
client. Service $\texttt{buy}$, expects to receive credit card details
on its private channel $k$ for a payment, and then spawns some service
$\texttt{serv}$ which sends some value (5 in this case). However, the
client expects to pay after it has used the service i.e. it will
invoke immediately $\text{serv}$ after invoking
$\texttt{buy}$. Payment is done eventually after the service is used.
Notice that, if run on its own, the process will get stuck. However,
if provided with the right context, it progresses simply because it
will reduce in the presence of another suitable service
$\texttt{s}$. Again, in \cite{DLY07}, this term is rejected because of
the inconsistent order in which $\texttt{buy}$ and $\texttt{s}$ are
used. In this work, this process is \stable, and thus has progress.

Other approaches to progress for session types include
\cite{sessionMariangiola,CDY07}, which considers models featuring
synchronous/asynchronous session types for object-oriented
languages. In particular, \cite{sessionMariangiola} also exploits the
assumption that a service cannot rely on already open
communications. However, in that work, a delegation $\outputS k
{k'}\pfx P$ is well-typed only if $P$ does not contain further uses of
$k$. That is, a process can truly only ever participate in one session
at a given time.  In contrast, while we insist that services at the
outset do not reference open sessions, a service may evolve through
delegation to participate in many sessions at the same time. A case in
point is (\ref{eq:2}) above.

In \cite{BCDDDY:concur2008}, a typing system ensures progress for
multiparty session types where sessions may involve more than two
participants. Unlike our result, the aforementioned works introduce
extra typing for guaranteeing progress of programs. However, because
of asynchronous communications, limitations as the one described in
the previous paragraph are present only for processes receiving a
delegated channel.

The present paper expands in several directions the session dependency
graphs introduced in \cite{CHY:concur2008}. First, the graphs in
\cite{CHY:concur2008} address a different language dealing with
exceptions and do not handle higher-order sessions (delegation).
Second, whereas \cite{CHY:concur2008} treated only programs, the
present work takes the much larger class of \stable\ processes.

\paragraph{Contribution of the paper.} The main technical contribution
of this paper is the investigation of session interdependency graphs
and the identification of a class of well-typed processes with
progress, the \emph{\stable} processes. This is an interesting class
because:
\begin{enumerate}
\item it is potentially practically relevant: as explained above, it
  is a natural model of WS;
\item it is simply characterized by acyclicity of its session
  dependency graph, a syntactic condition, checkable in linear
  time. No special typing system is necessary;
\item it includes new processes not identified as progressing by known
  methods.
\end{enumerate}
This contribution is relevant to programming languages based on
session types, e.g., \cite{HU07TYPE-SAFE}, as it gives a simple
syntactic means for ensuring progress for protocols implemented in
such languages.

\paragraph{Outline of the presentation.} We present the $\pi$-calculus
with sessions in \S~\ref{sec:calculus} and introduce a slight
variation of session types in \S~\ref{sec:types}.  \S~\ref{sec:graphs}
introduces session dependency graphs while \S~\ref{sec:stable} defines
the class of \stable\ processes, and prove that this class is closed
under reduction. \S~\ref{sec:progress} proves that every \stable\
process progresses, and, as a corollary, that so does every
program. We conclude in \S~\ref{sec:conclusions}.
For space reasons, some (parts of) proofs have been omitted and moved
\ifapp to the appendix.  \else to the appendix of the online version
\cite{onlineversion}.  \fi


\section{A $\pi$-calculus with Service Oriented Sessions}
\label{sec:calculus}
We introduce a variant of the $\pi$-calculus with sessions
\cite{THK,honda.vasconcelos.kubo:language-primitives} which outlaws
restriction on public channels, and allows replicated behaviour only
for services.

\paragraph{Syntax.}
Let $a,b,c,x,y,z,\ldots$ range over {\em service (or public)
  channels}; $k,k',t,s,\ldots$ over {\em session (or private)
  channels}; and $e, e',\ldots$ over public channels, and arithmetic
and other first-order {\em expressions}. 
\begin{center}
  \fbox{  
    $
    \begin{array}{rllll}
      P::=&\phantom{{}\mid{}}
            \INACT 	                          &\text{(inact)}\qquad
      &\mid P\pp Q                                &\text{(par)}\\
      &\mid \new k P                              &\text{(resSess)}\qquad
      &\mid \gamma                                &\text{(prefix)}\\
      \\
      \gamma ::=
      &\phantom{\mid{}}\repInitIn{a}{k}\pfx P   &\text{(repServ)}\qquad
      &\mid \ifthenelse ePQ                       &\text{(cond)}\\
      &\mid \initIn{a}{k}\pfx P                  &\text{(serv)}\qquad
      &\mid \initOut{a}{k}\pfx P            &\text{(request)}\\
      
      &\mid \inputP{k}{x}\pfx P             &\text{(input)}\qquad
      &\mid \outputP{k}e\pfx P              &\text{(output)}\\
      &\mid \inputS{k}{k'}\pfx P            &\text{(inputS)}\qquad
      &\mid \outputS{k}{k'}\pfx P \qquad    &\text{(delegation)}\\
      &\mid \branch{k}{l} P\qquad           &\text{(branch)}\qquad
      &\mid \select{k}{l}\pfx P 	    &\text{(select)}
    \end{array}
    $
  }
\end{center}
\smallbreak \NI Above, the class of prefixes $\gamma$ includes
services (serv), replicated services (repServ), and service
invocations (request); as well as in-session communication (input,
output), receive and send of session channels (inputS, delegation),
and branching (branch, select).  The other operators are standard. The
free session (service) channels of a process $P$, denoted by $\fsc P$
($\fv P$), are defined as usual. For the sake of simplicity, we have
removed recursion. 
We conjecture that our results can be easily extended.


\begin{example}[Buyer-Seller protocol]\rm\label{ex1}
  We recall a variant of the Buyer-Seller protocol from
  \cite{carbone.honda.yoshida:esop07} where a buyer invokes a service
  \texttt{buy} at a seller for a quote about some product. In case of
  acceptance by the buyer, the seller will place the order by invoking
  a service \texttt{ship} at a shipper and forward credit card
  details. Finally, the shipper will send directly a confirmation to
  the buyer. Such a protocol can be described by the process
  $P_{\text{Buyer}}\ \pp\ P_{\text{Seller}}\ \pp\ P_{\text{Shipper}}$
  such that:
  \begin{displaymath}
    \begin{array}{rl}

      P_{\text{Buyer}}\quad \stackrel{\text{def}}=\quad&
      \initOut{\texttt{buy}}{k}\pfx\ \inputP
      k{x_{\texttt{quote}}}\pfx\ \ifthenelse{\ x_{\texttt{quote}}\leq
        100\ }{\ \select k{\texttt{ok}}\pfx\inputP
        k{x_{\texttt{conf}}}\pfx\INACT\ }{\ \select
        k{\texttt{stop}}\pfx\INACT\ }\\[2mm]

      P_{\text{Seller}}\quad \stackrel{\text{def}}=\quad&
      \repInitIn{\texttt{buy}}{k}\pfx\outputP k{\texttt{quote}}\pfx\ k
      \triangleright\{\
      \texttt{ok}:\initOut{\texttt{ship}}{k'}\pfx\outputS{k'}{k}\pfx\INACT,\quad
      \texttt{stop}:\INACT\
      \}\\[2mm]

      P_{\text{Shipper}}\quad \stackrel{\text{def}}=\quad&
      \repInitIn{\texttt{ship}}{k'}\pfx\inputS{k'}{k}\pfx\
      \outputP{k}{\texttt{conf}}\pfx\INACT
    \end{array}
  \end{displaymath}
\end{example}

\paragraph{Structural Congruence and Reduction Semantics.}
The structural congruence $\equiv$ is standard and is defined as the
minimal relation satisfying the following rules:
\begin{center}
  \fbox{
    $
    \begin{array}{c}
      (i)\ P\pp Q\equiv Q\pp P \qquad\qquad (ii)\ P\pp(Q\pp R)\equiv
      (P\pp Q)\pp R\qquad\qquad (iii)\ P\pp\INACT \equiv P
      \\[2mm]
      (iv)\ P\equiv Q\quad(\text{if } P=_\alpha Q)
      \qquad\quad
      (v)\ P\pp \new k Q\equiv \new k(P\pp Q)\quad(k\not\in\fsc
      P)
      \qquad (vi)\ (\nu k) \INACT \equiv \INACT
    \end{array}
    $
  }
\end{center}
\smallskip The standard reduction semantics $\rightarrow$
\cite{honda.vasconcelos.kubo:language-primitives} is reported in
Table~\ref{reductionsemantics} where $e\Downarrow v$, taking
expressions to some values, is unspecified.  Note that we have adopted
the original \Did{Del} rule
\cite{honda.vasconcelos.kubo:language-primitives} which requires the
receiving side to ``guess'' what is being delegated (as in internal
$\pi$-calculus \cite{internalmobility}). As a consequence, process
$\inputS k{k''}\pfx P\pp\outputS k{k'}\pfx Q$ fails to reduce if $k'$
is free in $P$, as we cannot rename $k''$ to $k'$
\cite{sessiontypesRevisited}. We shall see that such process violates
transparency (see Remark~\ref{selfdelegation}).
\begin{table}[t]
  {\small
    \begin{center}
      \fbox{
        $
        \begin{array}{rl@{\hspace{0cm}}rl}
          \Did{RInit} \
          &
          \repInitIn a k\pfx P \pp \initOut a k\pfx Q
          \to
          \repInitIn a k\pfx P \pp \new k (P\pp Q)
          &
          \Did{Init} \
          &
          \initIn a k\pfx P \pp \initOut a k\pfx Q
          \to
          \new k (P\pp Q)
          \\[1.5mm]
          \Did{Com}\
          &
          \inputP{k}{x}\pfx P\pp \outputP{k}e\pfx Q
          \to
          P[v/x]\pp Q\qquad(e\Downarrow v)
          &
          \Did{Del}\
          &
          \inputS{k}{k'}\pfx P\pp \outputS{k}{k'}\pfx Q
          \to
          P\pp Q
          \\[1.5mm]
          \Did{Sel}\
          &
          \branch{k}{l} P\pp \select{k}{l_j}\pfx Q
          \to
          P_j\pp Q\ \ \ \quad(j\in I)\qquad\ \
          &
          \Did{Par}\
          &
          P
          \to
          P'
          \quad\Rightarrow\quad
          P\pp Q
          \to
          P'\pp Q
          \\[1.5mm]
          \Did{Res
          }\
          &
          P
          \to
          P'
          \quad\Rightarrow\quad
          \new k P
          \to
          \new k P'
          &
          \Did{Str}\
          &
          P\equiv Q,\
          Q\to Q',\
          Q'\equiv P'
          \ \Rightarrow\
          P
          \to
          P'
          \\[1.5mm]
          \Did{IfT}\
          &
          \ifthenelse ePQ
          \to
          P
          \qquad(e\Downarrow \true)
          &
          \Did{IfF}\
          &
          \ifthenelse ePQ
          \to
          Q
          \qquad(e\Downarrow \false)
        \end{array}
        $
      }
    \end{center}
  }
\caption{Reduction Semantics}
\label{reductionsemantics}
\end{table}

We conclude this section with two auxiliary notions. First, the notion
of program, i.e., a process containing no free session channels and no
occurrences of restricted session channels. 
\begin{definition}[Program]
  A process $P$ is a {\em program} whenever $\fsc P=\emptyset$ and
  there exists $P'\equiv P$ s.t.~$P'$ has no syntactic sub-term $\new
  k Q$.
\end{definition}
\NI Second, sub-processes, i.e., sub-terms of a process:
\begin{definition}[Sub-Process]
  A process $Q$ is a \emph{sub-process} of $P$ iff $Q$ is a sub-term
  of some $P'\equiv P$.  
\end{definition}


\section{Session Typing}
\label{sec:types}
\paragraph{Syntax.}
Session types abstract the way a single session channel is used within
a single session. The structure of a session is represented by a type,
which is then used as a basis for validating protocols through an
associated type discipline.  Their syntax is given by the following
grammar:
\begin{center}
  \fbox{  
    $
    \begin{array}{rlrlrlrlrl}
      \alpha&\ ::=\ 
      \inputT\theta\pfx\alpha
      \ \mid\ 
      \outputT\theta\pfx\alpha
      \ \mid\ 
      \branchT l\alpha
      \ \mid\
      \selectT l\alpha
      \ \mid\
      \inactT
      \\[1mm]
      \theta&\ ::=\ 
      S
      \ \mid\ 
      \alpha
      \qquad\qquad\qquad
      S\ ::=\ 
      \mathsf{basic}
      \ \mid\ 
      \langle\alpha\rangle
    \end{array}
    $
  }
\end{center}
\smallskip 

\NI Here, $\inputT\theta\pfx\alpha$ and $\outputT\theta\pfx\alpha$
denote in-session input and output followed by the communications in
$\alpha$. The type $\theta$ abstracts what is communicated: a basic
value ($\mathsf{basic}$ denotes basic types, e.g., $\mathsf{int}$ or
$\mathsf{bool}$), a service channel of type $\langle\alpha\rangle$, or
a session channel of type $\alpha$. Finally, $\branchT l\alpha$ and
$\selectT l\alpha$ denote branching and selection types, and $\inactT$
is the inactive session. The {\em dual} of $\alpha$, written
$\ol\alpha$, is defined as \smallskip
\begin{center}
  \fbox{
    $
    \begin{array}{rclrclrcl}
      \ol{\inputT\theta\pfx\alpha} & = & \outputT\theta\pfx\ol{\alpha}
      &
      \ol{\outputT\theta\pfx\alpha}& = & \inputT\theta\pfx\ol{\alpha}
      \\[1mm]
      \ol{\branchT l\alpha}        & = & \selectT l{\ol{\alpha}}
      &\qquad
      \ol{\selectT l\alpha}        & = & \branchT l{\ol{\alpha}}
      &\qquad
      \ol\inactT                   & = & \inactT
    \end{array}
    $
  }
\end{center}

\paragraph{Environments, Judgements and Typing Rules.}
We define two typing environments, namely the {\em service} and the
{\em session} typing.
\begin{center}
  \fbox{
    $
    \begin{array}{rlrlrlrl}
      \textit{(Service Typing)}\qquad\quad \Gamma &\ ::=\quad
      \Gamma,\;a:S\quad \mid\quad \emptyset
      \\[1mm]
      \textit{(Session Typing)}\qquad\quad \Delta &\ ::=\quad
      \Delta\;\cdot\;k:\alpha \quad\mid\quad \Delta\;\cdot\;k:\perp \quad
      \mid\quad \emptyset
    \end{array}
    $
  }
\end{center}
The service typing fixes usage of service channels, whereas session
typing fixes usage of session channels.  In $\Delta$, a session
channel $k$ may be assigned to $\perp$ rather than a session type
$\alpha$. This is to note that the two sides of session $k$ have
already been found in a sub-term (therefore $k$ cannot be used
further). When convenient, we treat both environments as functions
mapping channels to their types.

Judgements have the form $\tproves\Gamma P\Delta$ and are defined in
Table~\ref{table:typingrules}.
\begin{table}[t]
  \begin{center}\small
    \fbox{
      $
      \begin{array}{rl@{\hspace{0cm}}rl}
        \Did{T-Serv} \
        &
        \Rule
        {
          \tproves{\Gamma,\;a:\langle\alpha\rangle}{P}{k:\alpha}
        }
        {
          \tproves{\Gamma,\;a:\langle\alpha\rangle}{\initIn ak\pfx P}{\emptyset}
        }                                        
        &
        \Did{T-Req} \
        &
        \Rule
        {
          \tproves{\Gamma,\;a:\langle\alpha\rangle}{P}{\Delta\cdot
            k:\overline\alpha}
        }
        {
          \tproves{\Gamma,\;a:\langle\alpha\rangle}{\initOut ak\pfx P}{\Delta}
        }
       \\\\
        \Did{T-In} \
        &
        \Rule
        {
          \tproves{\Gamma,\;x:S}{P}{\Delta\cdot k:\alpha}
        }
        {
          \tproves{\Gamma}{\inputP kx\pfx P}{\Delta\cdot k:\;\inputT S\pfx\alpha}
        }                                        
        &
        \Did{T-Out} \
        &
        \Rule
        {
          \tproves{\Gamma}{P}{\Delta\cdot k:\alpha}
          \qquad
          \Gamma\vdash e:S
        }
        {
          \tproves{\Gamma}{\outputP ke\pfx P}{\Delta\cdot k:\;\outputT S\pfx\alpha}
        }
        \\\\
        \Did{T-InS} \
        &
        \Rule
        {
          \tproves{\Gamma}{P}{\Delta\cdot k:\alpha\cdot k':\beta}
        }
        {
          \tproves{\Gamma}{\inputS k{k'}\pfx P}{\Delta\cdot
            k:\;\inputT\beta\pfx\alpha}
        }                                        
        &
        \Did{T-Del} \
        &
        \Rule
        {
          \tproves{\Gamma}{P}{\Delta\cdot k:\alpha}
        }
        {
          \tproves{\Gamma}{\outputS k{k'}\pfx P}{\Delta\cdot k:\;\outputT\beta\pfx\alpha\cdot k':\beta}
        }                                        
        \\\\
        \Did{T-Bra} \
        &
        \Rule
        {
          \tproves{\Gamma}{P_i}{\Delta\cdot k:\alpha_i}
        }
        {
          \tproves{\Gamma}{\branch klP}{\Delta\cdot
            k:\;\branchT l\alpha}
        }\quad
        &
        \Did{T-Sel}\
        &
        \Rule
        {
          \tproves{\Gamma}{P}{\Delta\cdot k:\alpha_j}
        }
        {
          \tproves{\Gamma}{\select k{l_j}\pfx P}{\Delta\cdot
            k:\selectT l\alpha}
        }
        \\\\
        \Did{T-Par} \
        &
        \Rule
        {
          \tproves{\Gamma}{P_i}{\Delta_i}
          \qquad
          \Delta_1\asymp\Delta_2
        }
        {
          \tproves{\Gamma}{P_1\pp P_2}{\Delta_1\odot\Delta_2}
        }                                        
        &
        \Did{T-Inact} \
        &
        \Rule
        {
          \alpha_i=\inactT
        }
        {
          \tproves{\Gamma}{\INACT}{k_1:\alpha_1\cdot\ldots\cdot k_n:\alpha_n}
        }                                        
        \\\\
        \Did{T-RServ} \
        &
        \Rule
        {
          \tproves{\Gamma,\;a:\langle\alpha\rangle}{P}{k:\alpha}
        }
        {
          \tproves{\Gamma,\;a:\langle\alpha\rangle}{\repInitIn ak\pfx P}{\emptyset}
        }                                        
        &
        \Did{T-Res
        } \
        &
        \Rule
        {
          \tproves\Gamma P{\Delta\cdot k:\;\perp}
        }
        {
          \tproves{\Gamma}{\new kP}{\Delta}
        }
        \\\\
        \Did{T-Cond} \
        &
        \Rule
        {
          \Gamma\vdash e:\textsf{bool}
          \quad
          \tproves\Gamma {P_i}\Delta
        }
        {
          \tproves{\Gamma}{\ifthenelse e{P_1}{P_2}}{\Delta}
        }
        &
        \Did{T-Bot} \
        &
        \Rule
        {
          \tproves\Gamma {P}\Delta\cdot k:\inactT
        }
        {
          \tproves\Gamma {P}\Delta\cdot k:\perp
        }
      \end{array}
      $
    }
  \end{center}
  \caption{Typing Rules for Session Types}
  \label{table:typingrules}
\end{table}
The rules are all standard
\cite{honda.vasconcelos.kubo:language-primitives} except from
\Did{T-Serv}. As mentioned in the introduction, the fresh channel $k$
is the \emph{only} session channel present in the session environment
for the subprocess $P$. \Did{T-Bot} is necessary in order to guarantee
subject congruence \cite{sessiontypesRevisited}. We remind the reader
that, in \Did{T-Par}, $\Delta_1\asymp\Delta_2$ (duality check) holds
whenever $\Delta_1(k)=\overline{\Delta_2(k)}$ for every
$k\in\textsf{dom}(\Delta_1)\cap\textsf{dom}(\Delta_2)$.  Moreover,
$\Delta_1\odot\Delta_2$ (assign $\perp$ when both sides of a session
have been found) is defined as
$\Delta_1+\Delta_2+\bigcup_{k\in\textsf{dom}(\Delta_1)\cap\textsf{dom}(\Delta_2)}k:\perp$.


\begin{example}\rm\label{asdfjahsf}
  The Buyer-Seller protocol from Example~\ref{ex1} is clearly well
  typed according to the rules in Table~\ref{table:typingrules}. In
  fact, the judgement $\tproves\Gamma{P_{\text{Buyer}}\ \pp\
    P_{\text{Seller}}\ \pp\ P_{\text{Shipper}}}{\emptyset}$ holds
  whenever $\Gamma$ contains:
  \begin{displaymath}
    \begin{array}{rl}
      \texttt{buy}\ : \ &\outputT{\texttt{int}}\pfx\&\{\
      \texttt{ok}:\outputT{\texttt{string}}\pfx\inactT,\quad
      \texttt{stop}:\inactT\ \}\\
      \texttt{ship}\ :\ &\inputT{\outputT{\texttt{string}}}\pfx\inactT
    \end{array}
  \end{displaymath}
\end{example}

Clearly, a process well-typed according to
Table~\ref{table:typingrules} is also well-typed according to the
standard typing \cite{honda.vasconcelos.kubo:language-primitives}.  As
a consequence, it is straightforward to obtain the standard subject
congruence/reduction results.

\begin{theorem}\label{thm:std}
  Let $\tproves\Gamma P\Delta$. Then,
  (1) $P\equiv Q$ implies $\tproves\Gamma Q\Delta$; and (2) $P\to P'$
  implies $\tproves\Gamma P\Delta$.
\end{theorem}


\section{Session Dependency Graphs}
\label{sec:graphs}
In this  section, we recall, generalize, and develop \emph{session
  dependency graphs}, introduced in \cite{CHY:concur2008} for a
language dealing with exceptions and without delegation. In subsequent
sections, we shall use them to establish progress for a class of
well-typed processes.

Informally, given a restriction-free process
$P\equiv\gamma_1\pp\cdots\pp \gamma_n$ where each $\gamma_i$ is called
a ``thread'', $P$'s session dependency graph has a node for each
thread $\gamma_i$, and edges between threads $\gamma_i,\gamma_j$ if the two
share a free session channel.
Observe that a thread $\gamma_i$ can be blocked on a session channel
$k$, waiting for $\gamma_j$ to synchronize on $k$, \emph{only if} the
two have an edge between them. It follows that a well-typed process is
deadlocked only if its session dependency graph contains a cycle (we
prove this result formally in
Theorem~\ref{thm:progressstableprocesses}). Notice how this approach
to deadlock detection differs from the focus on the ordering of the
channels usage found frequently in the literature, e.g., in
\cite{Kobayashi06,DLY07}. 

\begin{definition}[Session Dependency
  Graph]\label{def:sessiondependencygraph}
  Let $P$ be well-typed process. The \emph{session dependency graph}
  $\graph P=(\nodes P, \edges P, \labels P)$ is the labelled
  unoriented graph\footnote{Technically, a graph is here a 5-tuple
    $(N, E, L, \dom, \cod)$, where the latter two are maps
    $\dom,\cod:E\to N$ taking edges to their domain and co-domain
    respectively. We elide these maps; their values shall always be
    clear from the context.  } with nodes $\nodes P$, edges $\edges P$
  and labels $\labels P:\nodes P \to \mathcal P(\sessionnames)$
  defined inductively on the term structure of $P$ as follows:
  \smallskip
  \begin{center}
    \fbox{ $
      \begin{array}{rclllll}
        \graph \INACT &=&\ (\emptyset,\quad \emptyset,\quad \emptyset)\\[1.5mm]
        \graph {\new k P} &=&\ (\ \nodes P,\quad \edges
        P,\quad \labels P\setminus\{k\}\ )\\[1.5mm]
        \graph{\gamma} &=&\
        (\ \bullet,\quad \emptyset,\;
        [\bullet\mapsto\fsc{\gamma}]\ )\\[2mm]
        \graph {P\pp Q} &=&\ (\
        \;\nodes P  + \nodes Q,\ \quad 
        \edges P + \edges Q +
        \displaystyle\Sigma_{
          {\tiny
            \begin{array}{l}
              p\in\nodes P\\[0.01mm]
              q\in\nodes Q\\[0.1mm]
              k\in \labels P (p) \cap\labels Q (q)
            \end{array}
          }
        }
        \!\!\!\!\!\!\!\!\!\!\!\!\!\!(p,q),\
        \quad
        \labels P + \labels Q
        \ )
      \end{array}
      $
    }
  \end{center}
  \smallskip
  Here, we take $(\labels P \setminus \{k\})(p)=\labels
  P(p)\setminus\{k\}$.
\end{definition} 
The definition of graph is insensitive to structural congruence:
\begin{lemma}
  \label{lem:56}
  Let $P,Q$ be well-typed processes s.t.~$P\equiv Q$. Then $\graph P\simeq\graph Q$.
\end{lemma}
\begin{proof}
  By cases on the definition of $\equiv$. 
\ifapp
  See Appendix~\ref{app:lem:56} for details.
\fi
\end{proof}

\begin{example}\rm
  Consider the following process.
  \begin{equation}
    \label{eq:1}
    \inputP {k'} x\pfx \outputP {k''} x \pp
    \inputP {k''} x\pfx \outputP {k'} x
  \end{equation}
  The session dependency graph of this process has two nodes, both
  labelled by $\{k',k''\}$, and thus has two edges between those two
  nodes:
  \begin{center}
    \includegraphics{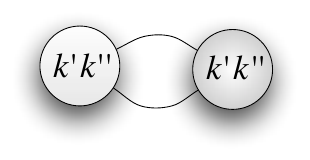}
  \end{center}
  Thus, the graph has a cycle.  Consider the same process, only
  restricting $k',k''$.
  \begin{equation}
    \label{eq:3}
    \new {k'k''}(\inputP {k'} x\pfx \outputP {k''} x \pp
    \inputP {k''} x\pfx \outputP {k'} x)
  \end{equation}
  Its graph has the same node and edges, but both nodes are now
  labelled by $\emptyset$. 
  \begin{center}
    \includegraphics{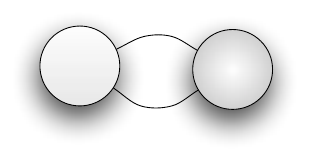}
  \end{center}
  \vspace{-5mm} Obviously, also this graph contains a cycle. Now,
  prefix the process with a service.
  \begin{equation}
    \label{eq:4}
    \initIn a k \pfx
    \new {k'k''}(\inputP {k'} x\pfx \outputP {k''} x \pp
    \inputP {k''} x\pfx \outputP {k'} x)
  \end{equation}
  This process has a graph with a single, unlabelled node and no
  edges which, obviously, contains no cycle:
  \begin{center}
    \includegraphics{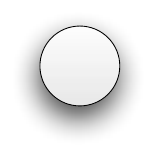}
  \end{center}
\end{example}

\begin{example}\rm
  The graph of $P_{\text{Buyer}}\ \pp\ P_{\text{Seller}}\ \pp\
  P_{\text{Shipper}}$ in Example~\ref{ex1} has clearly three
  unlabelled nodes. However, after one step of reduction, we have: 
\begin{equation}\label{theworldisevil}
\new k\Big(\
      \inputP k{x_{\texttt{quote}}}\pfx\ Q_1
      \ \pp\
      \outputP k{\texttt{quote}}\pfx\ Q_2
      \ \Big)
      \ \pp\
      P_{\text{Shipper}}
\end{equation}
where $Q_1$ and $Q_2$ are the remainders of buyer and seller's
processes. The session interdependency graph of the process above
becomes (restriction removes labels):
\begin{center}
  \includegraphics{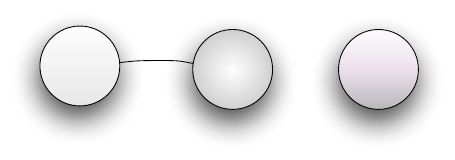}
\end{center}
If we further reduce (\ref{theworldisevil}) until also shipper is
invoked, we obtain the process:
\begin{equation*}
  \new {k,k'}\Big(\
  \inputP k{x_{\texttt{conf}}}\pfx\INACT
  \ \ \pp\ \
  \outputS{k'}{k}\pfx\INACT
  \ \ \pp\ \
  \inputS{k'}{k}\pfx \outputP{k}{\texttt{conf}}\pfx\INACT
  \ \Big)
\end{equation*}
whose graph is:\qquad\qquad\quad\qquad
\begin{tabular}{c}
  \includegraphics{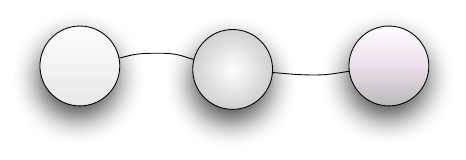}
\end{tabular}
\end{example}

\paragraph{Notation.}
To ease the presentation, we shall frequently say, e.g., ``$p$ is a
node of $\graph{P}$'' and ``$p$ is labelled by $k$ in $\graph{P}$''
rather than the less readable ``$p\in\nodes P$'' and ``$k\in\labels
P(p)$''.  Moreover, when no confusion may arise, we drop the ``in/of
$\graph P$'', saying simply, e.g., ``$p$ is labelled $k$''.


\smallskip

We conclude this section by noting two important facts about session
dependency graphs. First, the combinatorics of typing.
\begin{lemma}\label{lem:564}
  Let $\tproves{\Gamma}{P}{k_1:\perp\cdot\ldots\cdot k_n:\perp\cdot
    k_1':\alpha_1\cdot\ldots\cdot k_m':\alpha_m}$. Then (1) for each
  $i\leq m$ at most one node $p$ of $\graph P$ is labelled $k_i'$, and
  (2) when $R$ contains no top-level restrictions $\graph{P}$ has at
  most $n$ edges.
\end{lemma}
\begin{proof}
  By induction on $R$. 
\ifapp
  See Appendix~\ref{app:lem:564} for details.
\fi
\end{proof}

Second, the structure of programs.  Recall from
Section~\ref{sec:calculus} that programs have no free session channels
and no non-trivial restricted session channels.
\begin{proposition}
\label{prop:programs-are-stable}
Let $P$ be a well-typed program and $Q$ any of its sub-process. Then
$\graph Q$ has no edges.
\end{proposition}
\begin{proof}
  Note $\tproves\Gamma P\emptyset$ for some $\Gamma$.  Let $R$ be a
  sub-process of $P$. Then, also $\tproves{\Gamma_R}{R}{\Delta_R}$ for
  some $\Gamma_R,\Delta_R$.  By Lemma~\ref{lem:564}, it is enough to
  prove that every $k\in\dom(\Delta_R)$ has $\Delta_R(k)\not=\bot$. We
  show that $\Delta_R(k)=\perp$ for some $k$ would contradict
  $\tproves\Gamma P\emptyset$. But for this, it is sufficient to prove
  that every rule preserves $k:\bot$ in $\Delta$ from premises to
  conclusion. This is trivial for all rules but \Did{T-Res} and
  \Did{T-Serv}. Now, \Did{T-Res} does not apply to programs and their
  sub-processes, and \Did{T-Serv} accepts no $k:\bot$ in its premises.
\end{proof}

\section{\Stable\  Processes}
\label{sec:stable}
In this section, we define the notion of \emph{\stable\ process} and
investigate its properties. In particular, we prove that every program
is \stable\ and that \stable\ processes are closed under
reduction. Along the way, we get to thoroughly exercise session dependency graphs,
exhibiting their particular strengths. These results pave the way for
proving that every \stable\ process has progress in the next section.

\smallskip

We define the \stable\ processes as those whose graphs are,
essentially, ``everywhere acyclic''.
\begin{definition}[\Stable\  Process]
  Let $P$ be a well-typed process. We say that $P$ is \emph{\stable}
  iff every sub-process $Q$ of $P$ has $\graph{Q}$ acyclic.
\end{definition}

\begin{example}\rm
  Neither the process of (\ref{eq:1}), of (\ref{eq:3}) nor of
  (\ref{eq:4}) are \stable. The former two because they themselves
  have cyclic session dependency graphs, the latter because it
  contains a sub-process which has (namely (\ref{eq:3})). However, the
  buyer-seller system from Example~\ref{ex1} is transparent.
\end{example}


\Stable\ processes are closed under structural congruence:
\begin{lemma}
  \label{lem:5600}
  Let $P,Q$ be well-typed processes s.t.~$P\equiv Q$. Then $P$ \stable\  iff $Q$ \stable.
\end{lemma}
\begin{proof}
  By cases on the definition of $\equiv$. 
\ifapp
  See Appendix~\ref{app:lem:56} for details.
\fi
\end{proof}

The rest of this section is dedicated to proving our first main
result, i.e., \stability{} is preserved by reduction. In the next
section, we will use these results to prove that \stable\ processes
progress.

In order to see that a reduction $P\to P'$ preserves \stability, we
need to relate the session dependency graph of $P$ to the one of
$P'$. Informally, we make the following observations:
\begin{enumerate}
\item The set of free session channels of a process is non-increasing
  under reduction. Thus, even though the nodes and edges of
  $\graph{P}$ and $\graph{P'}$ might be very different, the {\em
    labels} of $\graph{P'}$ contain only names also in labels of
  $\graph{P}$.
\item We can speak of a label $k$ of $\graph{P}$ having a path to a
  label $k'$ in both graphs if there is a path from a node labelled
  $k$ to a node labelled $k'$. Thus, we arrive at the invariant: if
  $k$ has a path to $k'$ in $\graph{P'}$, it also had so in
  $\graph{P}$. This property is enough to ensure \stability{}
  preservation.
\end{enumerate}


We proceed to define precisely this notion of a path from $k$ to $k'$.

\begin{definition}
  Let $P$ be a well-typed process, and let $k,k'\in\fsc{P}$. We say
  that ``$k\leadsto k'$ in $\graph P$'' iff there exist nodes $p,p'$
  in $\graph P$, labelled by $k,k'$ respectively, s.t.~there is a path
  from $p$ to $p'$ in $\graph P$.
\end{definition}

Observe that the ``$-\leadsto-$'' relation is reflexive (because
$k,k'$ need not be distinct and there is always a path from a node to
itself) and symmetric (because the graph is). Moreover, because any
two nodes $p,p'$ sharing a label $k$ will have an edge between
them, we could equivalently have defined ``$k\leadsto k'$ iff
\emph{for any two} nodes $p,p'$ labelled by $k,k'$, there is a path
from $p$ to $p'$.'' It follows that $- \leadsto -$ is 
transitive.

\begin{remark}\rm
  One of the key insights carrying the proof that reduction preserves
  \stability{} is that $-\leadsto -$ is non-increasing under reduction
  of transparent processes. It is instructive to see this for a
  particular case involving delegation. Consider the following
  reduction:
  \begin{displaymath}
      \underbrace{\outputP {k} 5}_{\gamma_1}
      \ \pp\
      \underbrace{\outputS  {k'} {k}\pfx \inputP {k'} x}_{\gamma_2}
      \ \pp\
      \underbrace{\inputS {k'} {k}\pfx \inputP {k} x \pfx
        \outputP{k'}7}_{\gamma_3}  
      \quad\to\quad
      \underbrace{\outputP {k} 5}_{\gamma_1}
      \ \pp\
      \underbrace{\inputP {k'} x}_{\gamma'_2}
      \ \pp\
      \underbrace{\inputP {k} x \pfx \outputP{k'}7}_{\gamma'_3}  
  \end{displaymath}
  Before reduction, $\gamma_1$ and $\gamma_2$ share session $k$, and
  $\gamma_2$ and $\gamma_3$ share the session $k'$. After reduction,
  session $k$ has moved, and is now between $\gamma_1$ and
  $\gamma_3$. Graphically, the session dependency graphs change as
  follows:
  \begin{center}
    \includegraphics{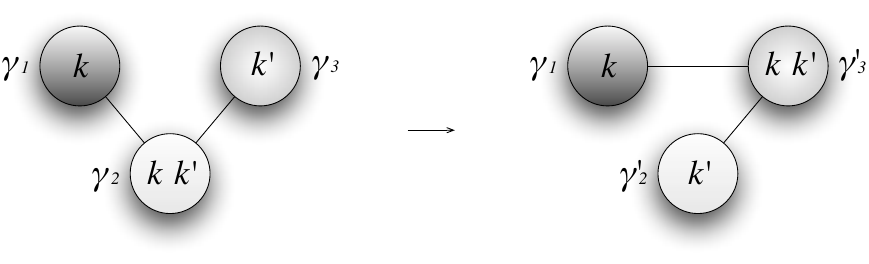}
  \end{center}
  However, as it is immediately obvious from the graphical
  representation, connectivity of $k$ and $k'$ did not change: they
  met at exactly one node before reduction, and they meet at exactly
  one node after reduction. This is the essential reason why
  delegation cannot introduce a cycle in the session dependency graph.
\end{remark}

\begin{theorem}\label{thm:stabilitypreservation}[Preservation of
  \stability]
  Let $R$ be a well-typed process s.t.~$R\to R'$. Then
  \begin{enumerate}
  \item If ~$\graph R$ is \stable\ then so is $\graph R'$, and
  \item if $k\not\leadsto k'$ in $R$, then also $k\not\leadsto k'$ in
    $R'$.
  \end{enumerate}
\end{theorem}
\noindent Before the proof, observe that part (2) is vacuously true for
      $k,k'\not\in\fsc{R}$.
      \begin{proof}[Proof (sketch).] Technically, the proof proceeds
        by induction on the derivation of $R\to R'$. Hereby, we
        discuss the most interesting cases. \ifapp All cases are reported in
        Appendix~\ref{app:thm:stabilitypreservation}.\fi
  \begin{itemize}

  \item \Did{Init}\quad$\initIn a k\pfx P \pp \initOut a k\pfx
    Q\to\new k (P\pp Q)$.
    \begin{enumerate}
    \item Because $\graph{\initIn a k\pfx P \pp \initOut a k\pfx Q}$
      is \stable, so are $\graph{P}$ and $\graph{Q}$.  Thus, in order
      to prove $\graph{\new k(P\pp Q)}$ \stable, it is sufficient to
      prove it acyclic. By \Did{T-Serv}, there exists $\alpha$
      s.t.~$\tproves {\Gamma} P {k:\alpha}$. Thus, by Lemma
      \ref{lem:564}, the nodes of $\graph{P}$ have empty labelling,
      except for at most one node; and that this unique node is
      labelled $k$ (if it exists at all). Now note that there exists
      also $\Delta'$ s.t.~ $\tproves {\Gamma} Q {\Delta'\cdot
        k:\overline\alpha}$. By linearity of the session environment,
      and Lemma \ref{lem:564}, we find that at most one node of
      $\graph Q$ is labelled $k$. Thus $\graph{\new k(P\pp Q)}$ is
      formed by adding at most one edge between acyclic graphs $\graph
      P$ and $\graph Q$ and is hence itself acyclic.

    \item It is sufficient to prove $k'\leadsto_R k''$ for any
      $k',k''\in\fsc{R}$. Observe that $\graph{\initIn a k\pfx P \pp
        \initOut a k\pfx Q}$ has exactly two nodes, one in
      $\graph{\initIn a k\pfx P}$ and one in $\graph{\initOut a k \pfx
        Q}$.  As $\fsc{\initIn a k\pfx P}$ is empty by typing, it
      follows that the free session channels of the $R$ are
      $\fsc{\initOut a k\pfx Q}$. But then the node $\graph{\initOut a
        k\pfx Q}$ is labelled by every free session channel name of
      $R$, whence trivially, $k'\leadsto_{R} k''$.
    \end{enumerate}

    \smallbreak



  \item \Did{Com}\quad$\inputP{k}{x}\pfx P\pp \outputP{k}e\pfx Q \to
    P[v/x]\pp Q\quad\qquad(e\Downarrow v)$.
    \begin{enumerate}
    \item As $\graph {R}=\graph{\inputP{k}{x}\pfx P\pp
        \outputP{k}e\pfx Q}$ is \stable\ then also $\graph P[v/x]$ and
      $\graph Q$ are \stable\ (the former because neither $v$ nor $x$
      are session channel names). It is now sufficient to prove that
      the graph $\graph {R'}=\graph{P[v/x]\pp Q}$ is acyclic.  Now,
      assume for a contradiction that $\graph{R'}$ contains a
      cycle. Because $\graph P$ and $\graph Q$ are \stable\ and thus
      acyclic, it must be the case that there exist two distinct edges
      connecting $\graph P$ and $\graph Q$.  Suppose wlog that these
      edges arise from $k',k''$ with $k',k''\in\fsc{P[v/x]}\cap
      \fsc{Q}$, and $k'\leadsto_P k''\leadsto_Q k'$. As $x$ is not a
      session channel name, so also $k',k''\in\fsc{P}$. But then also
      $k',k''\in\fsc{\inputP{k}{x}\pfx P}$ and
      $k',k''\in\fsc{\outputP{k}e\pfx Q}$, so a cycle is already
      in $\graph{R}$. Contradiction.
    \item We prove again that $k'\leadsto k''$ in $\graph{R}$ for any
      $k',k''\in\fsc{R}$.  Observe again that $\graph{R}=\graph{
        \inputP{k}{x}\pfx P\pp \outputP{k}e\pfx Q}$ has exactly two
      nodes, this time with an edge between them induced by $k$. But
      then $\graph R$ is in fact connected, so trivially $k'\leadsto
      k''$ in $\graph R$.
    \end{enumerate}
  \item \Did{Del}\quad $\inputS{k}{k'}\pfx P\pp \outputS{k}{k'}\pfx Q \to
     P\pp Q$
     \begin{enumerate}
     \item By assumption $R=\inputS{k}{k'}\pfx P\pp
       \outputS{k}{k'}\pfx Q$ \stable, so also $P,Q$ is \stable. It is
       now sufficient to prove that $P\pp Q$ is itself
       acyclic. Suppose it is not. Both $P,Q$ \stable\ and hence
       acyclic, so $P\pp Q$ cyclic must mean that there exists
       $k'',k'''\in\fsc{P}\cap\fsc{Q}$. We must have, $k'=k''$ or
       $k'=k'''$, or $R$ is itself cyclic; assume $k'=k'''$.  Thus
       $k',k''\in\fsc{P}\cap\fsc{Q}$. But because $R$ is well-typed,
       we must have $k'\not\in\fsc Q$; contradiction.
     \item Identical to the \Did{Com} case.
     \end{enumerate}


   \item \Did{Par}\quad $P \to P' \quad\Rightarrow\quad P\pp Q \to
     P'\pp Q$.
     \begin{enumerate}
     \item Because $P\pp Q$ \stable, also $P,Q$ \stable. By induction
       hypothesis, also $P'$ \stable, so it is sufficient to prove
       $P'\pp Q$ acyclic.  Assume for a contradiction that $P'\pp Q$
       contains a cycle. As $P',Q$ both acyclic, there must exist two
       edges between $\graph {P'}$ and $\graph Q$. Assume wlog that
       these edges are induced by distinct
       $k,k'\in\fsc{P'}\cap\fsc{Q}$ with $k\leadsto_{P'} k' \leadsto_Q
       k$. Because $P\to P'$ implies $\fsc{P'}\subseteq\fsc{P}$, we
       find $k,k'\in\fsc{P}$, so by induction hypothesis, $k\leadsto_P
       k'$, and by composition $k\leadsto_P k'\leadsto_Q k$, and
        a cycle is already in $P\pp Q$. Contradiction.

     \item We prove the contrapositive: Supposing $k\leadsto_{P'\pp Q}
       k'$ for $k,k\in\fsc{P\pp Q}$, we will see 
       that $k\leadsto_{P\pp Q} k'$. There must  exist a sequence
       $k=k_0,\ldots,k_n=k'$ with each $i$ having either
       $k_i\leadsto_{P'} k_{i+1}$ or $k_i\leadsto_Q k_{i+1}$. By the induction
       hypothesis,
       each $k_i\leadsto_{P'} k_{i+1}$ implies $k_i\leadsto_P
       k_{i+1}$, so, by stringing the path back together and
       noting that $\leadsto$ is transitive, we find $k\leadsto_{P\pp Q}
       k'$.  \end{enumerate} 
     \end{itemize}
     \vskip-1.25\baselineskip
\end{proof}




\section{Progress}
\label{sec:progress}
In this section, we give our main technical results: that every
\stable\ process has progress; and that all programs are \stable.
Intuitively, a process has {\em progress} if it cannot reduce to a
process that is ``stuck''. We shall follow \cite{DLY07} in taking a
non-trivial process to be stuck if it is irreducible in any
context. Thus, we do not consider a process stuck if it needs a
service or a counter-party to an active session.  In the sequel, a
process contains no live session channels whenever all its session
channels are bound by (repServ) or (serv). 
%
%
%
%
Moreover, let $E[\cdot]$ denote a reduction context, i.e., $E[\cdot]\
::=\ \cdot\ \ \mid\ \ E[\cdot]\pp P\ \ \mid\ \ \new k E[\cdot]$.
\begin{definition}[Progress]
  A process $P$ has 
  progress if $P\to^*P'$ implies that
  for every reduction context $E[\cdot]$ with $P'\equiv E[P'']$,
  whenever $P''$ contains live channels then there exists a process
  $Q$ s.t.
    \begin{center}
      \fbox{
        \begin{tabular}{rlrlllll}

          (a) & $Q\not\to$     &      
          (c) & $P''\pp Q\to R$
          \\[1mm]
          (b) & $P''\pp Q$ is well-typed   & \qquad(d) & $R$ has 
          progress\\[1mm]

        \end{tabular}
        }
    \end{center}
\end{definition}
\ \\
\NI Observe that processes with no live channels have progress.  




\smallskip

\NI This very strong but somewhat intentional definition of progress
captures the intuition that ``a stuck process is one where no thread
is permanently blocked''.  For a process $P$ to have progress, any
process $P'$ reachable from $P$ must be either without any live
channel or such that any of its top-level sub-processes $P''$ can be
composed in parallel with a process $Q$ and: (a) $Q$ is stuck; (b)
$P''\pp Q$ is well typed, i.e., $Q$ only provides services, requests
or counter-parties to active sessions in $P''$; and (c,d) $P''\pp Q$
can reduce to a process that also has progress.
The focus on sub-processes is to make sure that
non-terminating processes do not automatically have progress, e.g.,:
  \begin{equation}\label{equation123213}
   \inputP {k'} x\pfx \outputP {k''} x \ \pp\
   \inputP {k''} x\pfx \outputP {k'} x \ \pp\
   \repInitIn a k \pfx \initOut a k\ \pp\ \initOut a k
  \end{equation}



\begin{remark}\rm
  The definition above differs somewhat from the ones found in the
  literature, e.g., \cite{DLY07,BCDDDY:concur2008}. The present one has the distinct
  advantage that is independent of the means chosen to
  \emph{establish} progress.
  Other works use a special typing system to \emph{establish} progress
  already in the \emph{definition} of progress itself.
  For instance, the processes (\ref{eq:2}) and (\ref{eq:2asd}) on
  page~\pageref{eq:2}, intuitively have progress: all of their
  sessions run to termination when given access to appropriate
  services. And, both processes have progress wrt our
  definition. However, neither has progress as defined in
  \cite{DLY07}, where the definition of progress is inextricably
  linked to the typing system guaranteeing it, and both of those
  processes are untypable in that system (our definition requires
  typability for guaranteeing in-session linearity but it does not
  require transparency). As a further
  example, 
  write $P_1$ for the $\inputP {k'}x\pfx \outputP {k''} x \pp \inputP
  {k''} x\pfx \outputP {k'} x$ and $P_2$ for the $\inputP {k'}x\pfx
  \outputP {k''} x \pp \outputP {k'}e\pfx \inputP {k''} x$ in
 \[
 k\triangleright\{\quad l_1: P_1 \quad l_2: P_2\quad\}\quad\pp\quad\select k {l_2}
 \]
 Again, this process intuitively has progress --- it can only run to
 termination ---, it is included by the present definition of
 progress, but not by the one of \cite{DLY07}. Note that the process
 above is not transparent.
\end{remark}

Having established that \stability\ is preserved by reduction, we
proceed to show that a well-typed, \stable\ process has progress. 
In particular, we need to show that we can build the additional process
$Q$ found in the definition of progress. Our idea is to do that by
exploiting the type of a given process. The next result shows that every
session type is in fact inhabited, i.e., from a session type $\alpha$
we can reconstruct a process which behaves exactly as specified by
$\alpha$.  In the sequel, we shall assume that all basic types are
inhabited. For convenience, we further assume that they are all
inhabited by the same value ``$1$''; however this assumption is easily
rendered unnecessary.
\begin{lemma}[Every session type is inhabited]
  \label{lem:inhabitation}
  For any session type $\alpha$ and session channel $k$, there exists
  a \stable\ process $\llbracket
  \alpha\rrbracket^k$ with $\tproves {} {\llbracket
    \alpha\rrbracket^k} {k:\alpha}$, defined in Figure \ref{fig:inhab}.
\begin{figure}
  \begin{center}
    \fbox{
    \begin{tabular}{rl}

    \Did{inVal}&\qquad
    $\llbracket\inputT{\mathsf{basic}}\pfx\alpha\rrbracket^k= \inputP k
    x\pfx\llbracket\alpha\rrbracket^k$
    \\[1mm]
    \Did{inServ}&\qquad
    $\llbracket\inputT{\langle\beta\rangle}\pfx\alpha\rrbracket^k= 
    \inputP k a\pfx\llbracket\alpha\rrbracket^k$
    \\[1mm]
    \Did{inSess}&\qquad
    $\llbracket\inputT{\beta}\pfx\alpha\rrbracket^k= 
    \inputS k {k'}\pfx (\llbracket \alpha\rrbracket^k\pp \llbracket \beta \rrbracket^{k'})$
    \\[1mm]
    \Did{outVal}&\qquad
    $\llbracket\outputT{\mathsf{basic}}\pfx\alpha\rrbracket^k=
    \outputP k 1\pfx\llbracket\alpha\rrbracket^k$ 
    \\[1mm]
    \Did{outServ}&\qquad
    $\llbracket\outputT{\langle\beta\rangle}\pfx\alpha\rrbracket^k=
    \outputP k a\pfx (\llbracket \alpha \rrbracket^k\pp
    \repInitIn a {k'} \pfx \llbracket \beta \rrbracket^{k'})$
    \\[1mm]
    \Did{outSess}&\qquad
    $\llbracket\outputT{\beta}\pfx\alpha\rrbracket^k= 
    \new {k'} (\outputS k {k'} \pp \llbracket \alpha \rrbracket^k
    \pp\llbracket \ol\beta\rrbracket^{k'}) $
    \\[1mm]
    \Did{inSum}&\qquad
    $\llbracket\branchT l\alpha\rrbracket^k =
    \branch k l {\llbracket \alpha {\rrbracket^k}}$
    \\[1mm]
    \Did{outSum}&\qquad
    $\llbracket \selectT l\alpha\rrbracket^k =
    \select k l_1 \pfx \llbracket {\alpha_1} \rrbracket^k$
    \\[1mm]
    \Did{end}&\qquad
    $\llbracket\inactT\rrbracket^k = \INACT$.\\[1mm]
  \end{tabular}
  }
  \end{center}
  \caption{Translation $\llbracket-\rrbracket^-$}
\label{fig:inhab}  
\end{figure}
\end{lemma}
 \smallskip
\NI In the above lemma, the first three cases 
reconstruct a process that inputs a value, a service channel and a
session respectively. \Did{outVal}, \Did{outServ} and \Did{outSess}
are their dual counterpart. \Did{inSum} and \Did{outSum} are external
and internal choice while \Did{end} generates the inactive process
from the end type.
\begin{example}\rm
Here is an illustrative example. 
\[
  \llbracket
  \inputT{\;
    \inputT{\mathsf{basic}}\pfx
    \outputT{\mathsf{basic}}\;}
  \pfx\outputT{\mathsf{basic}}\rrbracket^k = 
  \inputS k {k'}\pfx (
    \outputP k 1
  \pp
    \inputP {k'} x\pfx
    \outputP {k'} 1
  )
\]
Note how received values are never used and how the process only ever
sends ``$1$''.
\end{example}
\begin{example}\rm
  In the buyer-seller protocol, the session type of the session
  \texttt{buy} reported in Example~\ref{asdfjahsf} is
  $\outputT{\texttt{int}}\pfx\&\{\
  \texttt{ok}:\outputT{\texttt{string}}\pfx\inactT,\quad
  \texttt{stop}:\inactT\ \}$.
  According to the construction from the previous Lemma, for some $k$,
  we can build the process $
  \outputP k{1}\pfx k\triangleright\{\
  \texttt{ok}:\outputP k{"a"}\pfx\INACT,\quad \texttt{stop}:\inactT\
  \}$, where we have chosen to inhabit \texttt{string} by the string
  $"a"$.
\end{example}
Knowing that every session type is inhabited, we can prove that
well-typed, \stable\ processes either have no live channels or can
reduce given the proper environment.

\begin{theorem}[Reduction of \Stable\  Processes]
\label{thm:progressstableprocesses}
Let $P$ be a well-typed, \stable{} 
process.  Either $P$ has no live channels, or there exists some
$Q\not\to$ such that $P\pp Q$ is well-typed, \stable{} 
and $P\pp Q \to$.
\end{theorem} 

\begin{proof}
  If $P$ has no live channels then we are done. Also, if $P\to$ then
  also $P\pp \INACT\to$, $\INACT\not\to$ trivially and clearly $P\pp
  \INACT$ \stable. So assume $P\not\to$. Assume wlog $P\equiv
  \new{\tilde k
  } (\gamma_1\pp\cdots \pp \gamma_n)$ and
  $\tproves \Gamma P\Delta$.

  Suppose first that for some $i$, $\gamma_i$ is a service invocation
  $\gamma_i\equiv \initOut a k\pfx Q$. 
  Then, by Lemma \ref{lem:inhabitation} there exists some $\Gamma'$
  with $\tproves {\Gamma,\Gamma'} {P\pp \initIn a k \pfx \llbracket
    \Gamma(a)\rrbracket^k} \Delta$; clearly $P\pp \initIn a k \pfx
  \llbracket \Gamma(a)\rrbracket^k$ is \stable\ and reduces.
  
  Suppose instead that no $\gamma_i$ is such a service invocation. Now
  consider the case where for some $k$, $k:\alpha\in\Delta$.  Clearly
  $k\not\in \tilde k$, so again by Lemma \ref{lem:inhabitation}, $P\pp
  \llbracket \overline\alpha \rrbracket^k$ is well-typed and \stable,
  and this process clearly reduces.  So, consider instead (and
  finally) the case where every $k$ mentioned by $\Delta$ in fact has
  $k:\bot\in\Delta$. We shall arrive at a contradiction, demonstrating
  that this typing is not possible for $P$ \stable\ with
  $P\not\to$. Observe first that, up to labels,
  $\graph{P}\simeq\graph{\gamma_1\pp\cdots\pp\gamma_n}$. We shall
  treat each $\gamma_i$ interchangeably as a process and as a node of
  this graph. Suppose wlog that no $\gamma_i$ has no live channels
  (otherwise, observe that $\gamma_i$ would contain no free session
  channels, whence we may conduct the following argument in the
  sub-graph of those $\gamma_i$ that have live channels). Thus, each
  $\gamma_i$ has an enabled action on some $k_i$. Because $P\not\to$
  and $P$ well-typed, the $k_i$ are pairwise distinct. Because $P$
  well-typed, for each $i$, there exists a unique $j$
  s.t.~$k_i\in\fsc{\gamma_j}$. But then $\graph{P}$ has $n$ nodes and
  at least $n$ edges, and must thus contain a cycle, contradicting
  \stability\ of $P$.
\end{proof}


\begin{theorem}
\label{thm:progress}
Well-typed and \stable{} 
processes have
progress.
\end{theorem}

\begin{proof}
  Suppose $P$ is well-typed and \stable{}.
  Moreover, let $P\to^*P'$ and $P'\equiv E[P'']$. Observe that by
  definition also $P''$ is well-typed and \stable{}. Now, if $P''$ has
  no live channels, we are done. Otherwise, by
  Theorem~\ref{thm:progressstableprocesses}, there exists $Q$
  s.t.~$P''\pp Q$ is \stable\ and well-typed, and $P''\pp Q\to R$ for
  some $R$. By Theorem \ref{thm:std}.2, this $R$ is also well-typed,
  so by Theorem~\ref{thm:stabilitypreservation}, $R$ is also \stable.
\end{proof}

\begin{remark}\label{selfdelegation}\rm
  In Section~\ref{sec:calculus}, we discussed that rule \Did{Del} can
  cause a deadlock with a processes of the form $\inputS k{k''}\pfx
  P\pp\outputS k{k'}\pfx Q$. However, this process is not transparent
  because its graph is not acyclic.
\end{remark}

Recall from Section \ref{sec:calculus} that a program is a process
with no free session channels and no non-trivial occurrences of session
channel restrictions.

\begin{corollary}
  \label{cor:program-progress}
  Well-typed programs have progress.
\end{corollary}
\begin{proof}
  By Proposition \ref{prop:programs-are-stable}, every well-typed
  program is a \stable\ process. By the preceding Theorem, every
  \stable\ process has progress.
\end{proof}

\begin{example}\rm
  Example~\ref{ex1}, buyer-seller,  is a
  program, and so transparent, whence it has progress.
\end{example}

\smallbreak 

We conclude by remarking on the complexity of the implied program
analysis. If a program is well-typed, transparency and thus progress
comes \emph{for free} courtesy of Corollary
\ref{cor:program-progress}. It can be checked whether a process is a
well-typed program in $\mathcal O(n)$, where $n$ measures the number
of nodes in the abstract syntax of the process; even if we have to
also ensure that the process contains no restrictions.

The broader class of transparent processes has progress by Theorem
\ref{thm:progressstableprocesses}. The session-dependency graph of a
process is computable in $\mathcal O(cp)$ where $c$ is the number of
distinct service-channels, bound or free, and $p$ is the number of
top-level prefixes. Deciding acyclicity is linear in the
number of nodes and edges, thus in time $\mathcal O(p+c)\subseteq
\mathcal O(pc)$. Transparency requires this computation at every sub-term
of the graph; however, it is clearly sufficient to consider every
maximal parallel product sub-term (i.e., for $P\pp Q$, no need to
consider  $P$, $Q$ separately). Thus, letting $p_i$ be
the width of the $i$th maximal parallel sub-term (that is, $n$ for
$P_1\pp\cdots\pp P_n$), we can compute transparency in $\mathcal
O(\Sigma_i p_ic)\subseteq\mathcal O(nc)$. In summary:
  \begin{theorem}
  A well-typed process $P$ can be checked for transparency in
    time $\mathcal O(nc)$, where $n$ is the number of nodes in the
    abstract syntax of the process, and $c$ is the number of distinct
    live channels. 
  \end{theorem}


\section{Conclusions}
\label{sec:conclusions}
We have provided a simple and {\em efficient} static analysis for
guaranteeing progress for web services based on session types.  The
advantage of our approach is that standard session typing (with the
restriction on the typing of services) is enough for guaranteeing
progress of programs without any further analysis of processes.  Our
result is based on the development of a technique which relies on
session interdependency graphs. In particular, we have shown that
transparent processes, those processes with an acyclic session
interdependency graph, have progress based on the fact that
transparency is preserved by the reduction semantics and it guarantees
that live channels eventually react.

\smallskip

The main limitation of this work is the lack of service channel
restriction. The main challenge with introducing such a syntactic
construct is in the definition of progress. In fact, processes such as
$\new a \big(\initOut a{k'}\pfx\inputP kx\ \pp\ \repInitIn a {k'}
\big)$ should satisfy the progress property. However, our current
definition would address the sub-process $\initOut a{k'}\pfx\inputP
kx$ which, taking restriction into account, has no progress. We leave
this issue as future work conjecturing that, under some small
assumptions, transparent processes with service restriction also have
progress.
Additionally, we plan to address two further points. Firstly, the
graph representation of session interdependency has proved to be a
very useful tool for investigating the properties of a system. We
believe that this approach can lead to further results that can go
beyond applications to progress e.g. secure data flow. Secondly, the
work in \cite{BCDDDY:concur2008} provides a typing system for
guaranteeing progress in multiparty sessions. A natural question is to
investigate whether the techniques used in this work can be reused in
the multiparty session setting with similar results. We are optimistic
that this might be the case.

\smallskip

\NI {\bf Acknowledgments.} We gratefully acknowledge helpful
discussions with N. Yoshida, M. Dezani, K.  Honda, and the anonymous referees.


\bibliographystyle{eptcs} 
\bibliography{session}

\begin{thebibliography}{10}
\providecommand{\bibitemstart}[1]{\bibitem{#1}}
\providecommand{\bibitemend}{}
\providecommand{\bibliographystart}{}
\providecommand{\bibliographyend}{}
\providecommand{\url}[1]{\texttt{#1}}
\providecommand{\urlprefix}{Available at }
\providecommand{\bibinfo}[2]{#2}
\bibliographystart

\bibitemstart{acciaiboreale2008}
\bibinfo{author}{Lucia Acciai} \& \bibinfo{author}{Michele Boreale}
  (\bibinfo{year}{2008}): \emph{\bibinfo{title}{A Type System for Client
  Progress in a Service-Oriented Calculus}}.
\newblock In: {\sl \bibinfo{booktitle}{Concurrency, Graphs and Models. Essays
  Dedicated to Ugo Montanari on the Occasion of His 65th Birthday}}, number
  \bibinfo{number}{5065} in \bibinfo{series}{LNCS},
  \bibinfo{publisher}{Springer}, pp. \bibinfo{pages}{642--658}.
\bibitemend

\bibitemstart{BCDDDY:concur2008}
\bibinfo{author}{Lorenzo Bettini}, \bibinfo{author}{Mario Coppo},
  \bibinfo{author}{Loris D'Antoni}, \bibinfo{author}{Marco~De Luca},
  \bibinfo{author}{Mariangiola Dezani-Ciancaglini} \& \bibinfo{author}{Nobuko
  Yoshida} (\bibinfo{year}{2008}): \emph{\bibinfo{title}{Global Progress in
  Dynamically Interleaved Multiparty Sessions}}.
\newblock In: {\sl \bibinfo{booktitle}{{19th International Conference on
  Concurrency Theory (Concur'08)}}}, \bibinfo{series}{LNCS},
  \bibinfo{publisher}{Springer}, pp. \bibinfo{pages}{418--433}.
\bibitemend

\bibitemstart{brunimezzina2008}
\bibinfo{author}{Roberto Bruni} \& \bibinfo{author}{Leonardo~Gaetano Mezzina}:
  \emph{\bibinfo{title}{A deadlock free type system for a calculus of services
  and sessions}}.
\newblock \bibinfo{note}{Draft available online}.
\bibitemend

\bibitemstart{caires:esop09}
\bibinfo{author}{Luis Caires} \& \bibinfo{author}{Hugo Vieira}
  (\bibinfo{year}{2009}): \emph{\bibinfo{title}{Conversation Types}}.
\newblock In: {\sl \bibinfo{booktitle}{18th European Symposium on Programming
  (ESOP'08)}}, {\sl \bibinfo{series}{LNCS}} \bibinfo{volume}{5502},
  \bibinfo{publisher}{Springer}, pp. \bibinfo{pages}{285--300}.
\bibitemend

\bibitemstart{onlineversion}
\bibinfo{author}{Marco Carbone} \& \bibinfo{author}{S\o ren Debois}
  (\bibinfo{year}{2010}): \emph{\bibinfo{title}{A Graphical Approach to
  Progress for Structured Communication in Web Services}}.
\newblock \bibinfo{note}{Online version of this paper available at
  \url{http://www.itu.dk/~maca/papers/CD10.pdf}}.
\bibitemend

\bibitemstart{carbone.honda.yoshida:esop07}
\bibinfo{author}{Marco Carbone}, \bibinfo{author}{Kohei Honda} \&
  \bibinfo{author}{Nobuko Yoshida} (\bibinfo{year}{2007}):
  \emph{\bibinfo{title}{{Structured Communication-Centred Programming for Web
  Services}}}.
\newblock In: {\sl \bibinfo{booktitle}{16th European Symp. on Programming
  (ESOP'07)}}, {\sl \bibinfo{series}{LNCS}} \bibinfo{volume}{4421},
  \bibinfo{publisher}{Springer}, pp. \bibinfo{pages}{2--17}.
\bibitemend

\bibitemstart{CHY:concur2008}
\bibinfo{author}{Marco Carbone}, \bibinfo{author}{Kohei Honda} \&
  \bibinfo{author}{Nobuko Yoshida} (\bibinfo{year}{2008}):
  \emph{\bibinfo{title}{Structured Interactional Exceptions for Session
  Types}}.
\newblock In: {\sl \bibinfo{booktitle}{{19th Int'l Conference on Concurrency
  Theory (Concur'08)}}}, \bibinfo{series}{LNCS}, \bibinfo{publisher}{Springer},
  pp. \bibinfo{pages}{402--417}.
\bibitemend

\bibitemstart{CDY07}
\bibinfo{author}{Mario Coppo}, \bibinfo{author}{Mariangiola Dezani-Ciancaglini}
  \& \bibinfo{author}{Nobuko Yoshida} (\bibinfo{year}{2007}):
  \emph{\bibinfo{title}{{Asynchronous Session Types and Progress for
  Object-Oriented Languages}}}.
\newblock In: {\sl \bibinfo{booktitle}{{FMOODS'07}}}, {\sl
  \bibinfo{series}{LNCS}} \bibinfo{volume}{4468}, pp. \bibinfo{pages}{1--31}.
\bibitemend

\bibitemstart{DLY07}
\bibinfo{author}{Mariangiola Dezani-Ciancaglini}, \bibinfo{author}{Ugo
  de'Liguoro} \& \bibinfo{author}{Nobuko Yoshida} (\bibinfo{year}{2007}):
  \emph{\bibinfo{title}{{On Progress for Structured Communications}}}.
\newblock In: {\sl \bibinfo{booktitle}{Trustworthy Global Computing (TGC'07)}},
  {\sl \bibinfo{series}{LNCS}} \bibinfo{volume}{4912},
  \bibinfo{publisher}{Springer}, pp. \bibinfo{pages}{257--275}.
\bibitemend

\bibitemstart{sessionMariangiola}
\bibinfo{author}{Mariangiola Dezani-Ciancaglini}, \bibinfo{author}{Dimitris
  Mostrous}, \bibinfo{author}{Nobuko Yoshida} \& \bibinfo{author}{Sophia
  Drossopoulou} (\bibinfo{year}{2006}): \emph{\bibinfo{title}{{Session Types
  for Object-Oriented Languages}}}.
\newblock In: {\sl \bibinfo{booktitle}{Proceedings of ECOOP'06}},
  \bibinfo{series}{LNCS}, \bibinfo{publisher}{Springer}, pp.
  \bibinfo{pages}{328--352}.
\bibitemend

\bibitemstart{honda.vasconcelos.kubo:language-primitives}
\bibinfo{author}{Kohei Honda}, \bibinfo{author}{Vasco~T. Vasconcelos} \&
  \bibinfo{author}{Makoto Kubo} (\bibinfo{year}{1998}):
  \emph{\bibinfo{title}{{Language Primitives and Type Disciplines for
  Structured Communication-based Programming}}}.
\newblock In: {\sl \bibinfo{booktitle}{7th European Symposium on Programming
  (ESOP'98)}}, {\sl \bibinfo{series}{LNCS}} \bibinfo{volume}{1381},
  \bibinfo{publisher}{Springer-Verlag}, pp. \bibinfo{pages}{22--138}.
\bibitemend

\bibitemstart{HU07TYPE-SAFE}
\bibinfo{author}{Raymond Hu}, \bibinfo{author}{Nobuko Yoshida} \&
  \bibinfo{author}{Kohei Honda} (\bibinfo{year}{2008}):
  \emph{\bibinfo{title}{Session-Based Distributed Programming in {J}ava}}.
\newblock In: \bibinfo{editor}{Jan Vitek}, editor: {\sl
  \bibinfo{booktitle}{ECOOP}},  \bibinfo{volume}{5142},
  \bibinfo{publisher}{Springer}, pp. \bibinfo{pages}{516--541}.
\bibitemend

\bibitemstart{Kobayashi06}
\bibinfo{author}{Naoki Kobayashi} (\bibinfo{year}{2006}):
  \emph{\bibinfo{title}{A New Type System for Deadlock-Free Processes}}.
\newblock In: {\sl \bibinfo{booktitle}{CONCUR'06}}, {\sl
  \bibinfo{series}{LNCS}} \bibinfo{volume}{4137}, pp.
  \bibinfo{pages}{233--247}.
\bibitemend

\bibitemstart{PucellaTov08}
\bibinfo{author}{Riccardo Pucella} \& \bibinfo{author}{Jesse Tov}
  (\bibinfo{year}{2008}): \emph{\bibinfo{title}{Haskell Session Types with
  (Almost) No Class}}.
\newblock In: {\sl \bibinfo{booktitle}{Proc. of the 1st ACM SIGPLAN Symposium
  on Haskell (Haskell'08)}}, \bibinfo{publisher}{ACM SIGPLAN}, pp.
  \bibinfo{pages}{25--36}.
\bibitemend

\bibitemstart{internalmobility}
\bibinfo{author}{Davide Sangiorgi} (\bibinfo{year}{1996}):
  \emph{\bibinfo{title}{pi-Calculus, Internal Mobility, and Agent-Passing
  Calculi}}.
\newblock {\sl \bibinfo{journal}{Theoretical Computer Science}}
  \bibinfo{volume}{167}(\bibinfo{number}{1{\&}2}), pp.
  \bibinfo{pages}{235--274}.
\bibitemend

\bibitemstart{scribble}
\bibinfo{author}{Scribble} (\bibinfo{year}{2008}).
\newblock \emph{\bibinfo{title}{{{S}cribble {P}roject}}}.
\newblock \bibinfo{note}{\url{www.scribble.org}}.
\bibitemend

\bibitemstart{THK}
\bibinfo{author}{Kaku Takeuchi}, \bibinfo{author}{Kohei Honda} \&
  \bibinfo{author}{Makoto Kubo} (\bibinfo{year}{1994}):
  \emph{\bibinfo{title}{{An Interaction-based Language and its Typing
  System}}}.
\newblock In: {\sl \bibinfo{booktitle}{PARLE'94}}, {\sl \bibinfo{series}{LNCS}}
  \bibinfo{volume}{817}, \bibinfo{publisher}{Springer-Verlag}, pp.
  \bibinfo{pages}{398--413}.
\bibitemend

\bibitemstart{W3C}
\bibinfo{author}{W3C}.
\newblock \emph{\bibinfo{title}{\mbox{W}orld-\mbox{W}ide \mbox{W}eb
  \mbox{C}onsortium}}.
\newblock \bibinfo{howpublished}{http://www.w3c.org}.
\bibitemend

\bibitemstart{sessiontypesRevisited}
\bibinfo{author}{Nobuko Yoshida} \& \bibinfo{author}{Vasco~Thudichum
  Vasconcelos} (\bibinfo{year}{2007}): \emph{\bibinfo{title}{Language
  Primitives and Type Discipline for Structured Communication-Based Programming
  Revisited: Two Systems for Higher-Order Session Communication}}.
\newblock {\sl \bibinfo{journal}{ENTCS}}
  \bibinfo{volume}{171}(\bibinfo{number}{4}), pp. \bibinfo{pages}{73--93}.
\bibitemend

\bibliographyend
\end{thebibliography}

\ifapp
\newpage
\appendix
\input{appendix}
\fi

\end{document}